\documentclass[aps,pra,reprint,superscriptaddress,nofootinbib,longbibliography]{revtex4-2}

\usepackage{graphicx}
\usepackage{multirow}
\usepackage{amsmath,amssymb,amsfonts}
\usepackage{amsthm}
\usepackage{mathrsfs}
\usepackage{xcolor}
\usepackage{textcomp}
\usepackage{booktabs}
\usepackage{array}
\usepackage{algorithmicx}
\usepackage{algpseudocode}
\usepackage{listings}
\usepackage{bm}
\usepackage{mathtools}
\usepackage[colorlinks=true,linkcolor=blue,citecolor=blue,urlcolor=blue]{hyperref}

\graphicspath{{figure/}}
\theoremstyle{plain}
\newtheorem{theorem}{Theorem}
\newtheorem{proposition}[theorem]{Proposition}
\theoremstyle{definition}

\theoremstyle{remark}

\newcounter{algorithmblock}
\renewcommand{\thealgorithmblock}{\arabic{algorithmblock}}
\newenvironment{algorithmblock}[2][]{%
	\refstepcounter{algorithmblock}%
	\par\medskip\noindent\begin{minipage}{\columnwidth}%
		\hrule\vspace{0.4em}%
		\noindent\textbf{Algorithm~\thealgorithmblock. #2}%
		\if\relax\detokenize{#1}\relax\else\label{#1}\fi%
		\vspace{0.4em}\hrule\vspace{0.4em}%
	}{%
		\vspace{0.4em}\hrule%
	\end{minipage}\par\medskip%
}

\begin{document}
	
	\title{Defect-Adaptive Lattice Surgery on Irregular Boundary Surface-Code Patches}
	
	\author{Gunsik Min}
	\email{mgs3351@korea.ac.kr}
	\affiliation{School of Electrical Engineering, Korea University, Seoul 02841, Republic of Korea}
	\author{Youshin Chung} \email{jeongys604@korea.ac.kr} \affiliation{School of Electrical Engineering, Korea University, Seoul 02841, Republic of Korea}
    
	\author{Yujin Kang}
	\email{yujin20@korea.ac.kr}
	\affiliation{School of Electrical Engineering, Korea University, Seoul 02841, Republic of Korea}
	
	\author{Jun Heo}
	\email{junheo@korea.ac.kr}
	\affiliation{School of Electrical Engineering, Korea University, Seoul 02841, Republic of Korea}

	\begin{abstract}
		Defect-adaptive surface-code methods construct valid logical patches on imperfect hardware, but fault-tolerant computation also requires certifying that requested logical operations remain executable on the resulting irregular geometries. We address this problem for lattice-surgery merges whose interfaces intersect deformed boundaries, unavailable checks, or gauge-inferred operators. Our compiler constructs an admissible seam-measurement set whose elements associate effective seam supports with raw-outcome selectors and schedule metadata. A GF(2) row-space test then determines whether the requested joint parity is realizable modulo pre-merge stabilizer constraints, returning either an explicit parity-extraction rule or a dual witness of non-realizability within the evaluated set. Paired ablations isolate the contribution of the proposed measurement-set construction rules. We further apply the test to configurations generated by Snakes and Ladders (SnL) on 15,000 clustered boundary-defect maps. Adding the proposed boundary-reconstruction rows increases the certification rate from 61.88\% to 62.79\%, with the gain concentrated in strongly clustered regimes. The framework therefore provides a compiler step that separates defect-adaptive patch construction from the certification and extraction of lattice-surgery parity measurements.
	\end{abstract}
	
	\maketitle
	
	\section{Introduction}
	
	Quantum error correction (QEC) is essential for scaling quantum processors from noisy few-qubit demonstrations to practical fault-tolerant computation~\cite{Gottesman1997}. Among the leading QEC architectures, the surface code is especially attractive because it requires only two-dimensional nearest-neighbor connectivity and exhibits a comparatively high threshold under realistic noise assumptions~\cite{Kitaev2003,Dennis2002,Fowler2012,Raussendorf2007,Wang2003}. In the ideal setting, increasing the code distance exponentially suppresses logical failure below threshold. In practice, however, large-scale processors must also tolerate fabrication defects and run-time component failures that violate the assumption of a static defect-free lattice.
	
	Recent experimental demonstrations have achieved below-threshold operation in small surface-code patches~\cite{GoogleAI2023,Krinner2022,Acharya2024}, making defect tolerance an increasingly practical concern. Because surface-code protection is geometric, defects can distort boundaries, disable local checks, alter the homology classes of logical operators, and reduce the minimum weight of an undetectable error if left untreated. Considerable effort has therefore been devoted to \emph{defect-adaptive} surface-code constructions that preserve a valid code space on irregular lattices by combining local deformation, gauge checks, and super-stabilizers~\cite{Stace2009,Bombin2010,Brown2017,Auger2017,Strikis2023,Nagayama2017,Heng2024}. Recent bandage-like constructions have further improved this picture by preserving more active data qubits, lowering super-stabilizer weight, and maintaining larger effective distance in clustered-defect regimes~\cite{BandageAdapter}.
	
	These advances largely solve a \emph{patch adaptation} or \emph{memory} problem: given an imperfect device, construct a valid logical surface-code patch that protects stored information with low overhead~\cite{CampbellTerhalVuillot2017}. Fault-tolerant computation, however, requires more than storing a logical qubit. It also requires a reliable way to perform logical operations repeatedly on those adapted patches.
	
	For the surface code, lattice surgery provides a standard low-overhead route to logical Clifford operations through merge and split operations and joint logical-parity measurements~\cite{Horsman2012,Litinski2019,FowlerGidney2018}. In the defect-free setting, the seam introduced during a merge activates a regular family of local checks whose product equals the desired joint parity, such as $Z_L \otimes Z_L$ or $X_L \otimes X_L$. On defect-adapted patches, this regular picture can fail. The seam may intersect a deformed boundary, some seam checks may be reduced or shifted, some effective checks may exist only through gauge inference~\cite{Bombin2015,Paetznick2013}, and the relevant information may be distributed across multiple measurement rounds by alternating or shell-based schedules. The key question is therefore no longer simply \emph{which seam is activated}, but rather: Given an irregular, defect-adapted merged patch, which available measurements realize the target joint logical parity?
	
	This work addresses that operation-level problem. We view defect-adaptive lattice surgery as a parity-synthesis task: given an irregular but valid defect-adapted patch, the compiler determines whether the available gauge and stabilizer measurements generate the requested logical parity and, when they do, returns an explicit rule for extracting it from the recorded outcomes. This viewpoint turns irregular lattice surgery from an ad hoc geometric repair problem into a certified synthesis problem.
	
	The main contributions of this paper are fourfold. First, we distinguish the existence of a valid defect-adapted patch from the availability of a requested lattice-surgery parity on that patch. Second, we define an \emph{admissible seam-measurement set}. Each element specifies an effective seam support, the raw-outcome selector used to reconstruct its value, and the associated source and schedule metadata. An element is included only if it is measurable or inferable from the available raw outcomes, commutes with the retained opposite-type constraints, and satisfies the stated local orientation rule. Third, we apply standard GF(2) row-space membership to this set. The calculation returns either seam and raw-outcome selectors or a dual witness showing that the target is not generated by the evaluated measurements and pre-merge constraints. The novelty therefore lies not in the binary linear-algebra test itself, but in constructing the physically valid measurement set and converting a successful solution into an explicit parity-extraction rule. Fourth, we evaluate the method on configurations generated by SnL using matched clustered boundary-defect maps. This comparison tests the proposed certification step on the measurement structures supplied by a published defect-adaptation method; it is not intended to establish a uniform performance ranking.

	The rest of the paper is organized as follows. Section~\ref{sec:background} reviews defect-adaptive patches, irregular lattice-surgery seams, and related circuit-construction methods. Section~\ref{sec:method} defines the admissible seam-measurement set, the GF(2) test, explicit parity extraction, and dual infeasibility certificates. Section~\ref{sec:simulations} presents the paired internal ablation and the published-method integration study. Section~\ref{sec:conclusion} summarizes the scope and limitations of the method.
	
	\section{Background and preliminaries}
	\label{sec:background}
	
	\subsection{Defect-adaptive surface-code patches}
	
	We consider a planar surface-code patch defined on a set $\mathcal{Q}$ of active data qubits together with commuting $X$-type and $Z$-type parity constraints. To avoid convention-dependent terminology, we label boundaries by the type of logical operator that can terminate there. A single logical patch therefore has two disjoint $X$-boundaries and two disjoint $Z$-boundaries, and logical operators are represented by nontrivial chains connecting the corresponding opposite boundaries.
	
	Defects are modeled as unavailable data qubits, syndrome qubits, or couplers. If untreated, such defects can invalidate local stabilizer measurements, distort the boundary geometry, and reduce logical distance. Defect-adaptive surface-code methods preserve a valid code space by modifying the local stabilizer description: unsafe checks are removed or replaced, boundaries are deformed, and higher-weight effective checks are introduced when necessary to maintain commutation and logical consistency~\cite{Nagayama2017,Heng2024,Barrett2010}.
	
	A key ingredient in modern defect adaptation is the use of \emph{super-stabilizers}, namely products of nearby same-type gauge or stabilizer fragments that remain measurable even when the original local checks are broken. In bandage-like constructions, these super-stabilizers are chosen to bridge defect regions while keeping support localized and weight as small as possible. This substantially improves the memory performance of irregular patches, but it does not by itself specify how logical merge and split operations should be carried out on those patches.
	
	\subsection{Gauge inference and schedule dependence}
	
	On irregular patches, the effective check logically needed by the code is often not measured directly as a single parity operator. Instead, one measures lower-weight gauge operators and infers the effective check from their product~\cite{Bombin2015,Paetznick2013}. This distinction is especially important near defect clusters, where local commutation constraints or hardware limitations make direct measurement of a desired high-weight operator impractical.
	
	In addition, irregular patches often impose nonuniform measurement schedules. For example, $X$-type and $Z$-type super-check families may be measured in alternating cycles, or shell-based schedules may repeat one gauge family for several consecutive rounds before switching type. As a result, the quantity needed for a logical merge is generally not the outcome of a single physical seam check in one round, but a post-processed parity constructed from time-tagged gauge outcomes.
	
	\subsection{Lattice surgery on irregular seams}
	
	In defect-free lattice surgery, a merge activates a regular family of seam checks and the product of those checks yields the desired joint logical parity. On defect-adapted patches, that statement is no longer automatic. A seam check may be absent, reduced, shifted by boundary deformation, bundled into a seam super-check, or recoverable only by gauge inference. The seam itself may no longer form a regular strip~\cite{Vuillot2019}.
	
	This motivates the central problem of the present work: given two defect-adapted patches and a requested lattice-surgery parity measurement, determine whether the desired joint logical parity can be realized by the seam-related measurements that are actually available. In this paper we focus on the merge primitive, which is the setting used in the numerical study, and construct the corresponding explicit parity-extraction rule when the requested parity is realizable.

	\subsection{Relation to prior defect-adaptive lattice-surgery work}
	\label{sec:related}
	
	The most closely related prior work is the defect-adaptive lattice-surgery framework
	of Leroux et al.~\cite{SnL}, hereafter SnL. SnL introduced ancilla-repurposing
	primitives for defective seam-check ancillas and showed how different repurposing
	orientations can either preserve or reduce the relevant code distance. This provides
	an important local mechanism for maintaining lattice-surgery measurements in the
	presence of defective measurement resources.
	
	The present work addresses a different, complementary layer of the problem. First, we consider seam-boundary data-qubit defects, where the data-qubit
	support of the intended seam check may itself be damaged. In this case, the issue is
	not only how to replace a missing measurement ancilla, but also how to decide which
	defect-adapted seam operator should represent the desired joint logical parity on the
	irregular geometry. This motivates the fused boundary seam-check construction and the
	opposite-type commutation tests illustrated in Fig.~\ref{fig:data_defect_merge}.
	
	Second, we use a GF(2) row-space layer as a compiler certification step after the
	admissible seam-measurement set has been constructed. Row-space membership over GF(2) is
	standard in stabilizer language; the point of the present construction is to determine
	which native, repaired, gauge-inferred, and fused seam measurements are physically available, and how they combine with the pre-merge constraints for an irregular
	lattice-surgery merge. When the row-space condition is satisfied, the same solution
	also specifies the explicit post-processing rule over the measured gauge outcomes.
	In this sense, the present framework should be viewed not as a replacement for local
	defect-repair primitives such as those in SnL, but as a compiler step that
	certifies and extracts the logical merge parity after such defect-adapted seam
	ingredients have been constructed.
	
	Recent detecting-region and dropout-based approaches provide another closely related
	line of work. In particular, LUCI~\cite{DebroyLUCI2024} and subsequent
	compilation frameworks such as ACID~\cite{WolanskiACID2025} and optimized
	measurement-schedule constructions~\cite{AnkerDebroy2025} construct or optimize
	fault-tolerant surface-code syndrome-extraction circuits in the presence of missing
	qubits, couplers, or measurement resources by modifying the circuit schedule and the
	associated detector structure. These methods address the circuit-realization problem
	for defective hardware and can preserve important distance properties for several
	classes of dropouts, at the cost of modified temporal structure or optimized
	measurement schedules.
	
	Our contribution is complementary to this line of work. Rather than optimizing the
	full syndrome-extraction circuit for a defective patch, we focus on the operation-level
	question that arises once a defect-adapted seam-measurement set is available: whether the
	intended lattice-surgery joint logical parity is generated by the native seam rows,
	repaired or gauge-inferred seam rows, fused seam super-checks, and separated
	pre-merge constraints. Detecting-region or dropout-based circuit constructions can
	therefore be viewed as possible sources of physically available measurements and detector
	information, while the present GF(2) synthesis layer certifies whether those rows
	realize the requested logical merge observable and returns the corresponding
	raw-outcome selector.
	
	Horsman et al.~\cite{Horsman2012}, Litinski~\cite{Litinski2019,LitinskiVonOppen2018},
	and Fowler and Gidney~\cite{FowlerGidney2018} establish the defect-free reference
	for lattice surgery on ideal patches. Landahl and Ryan-Anderson~\cite{Landahl2014}
	extend lattice surgery to color codes. Strikis et al.~\cite{Strikis2023} and Siegel
	et al.~\cite{Siegel2023} address boundary deformation and distance preservation in
	the memory setting. The present work connects these two lines by asking how the
	logical seam parity itself should be certified and extracted once the patch and seam
	geometry have become irregular.
	
	Operationally, local defect repairs and dropout-compiled detector structures can supply measurements to the proposed set. For example, ancilla repurposing provides an inferred seam measurement together with the raw gauge outcomes needed to reconstruct it. The measurement is included only if it passes the same commutation and implementation checks as the native seam measurements. The GF(2) step then determines whether the resulting set, together with the pre-merge constraints, generates the requested logical parity. Boundary data-qubit defects can additionally require a fused boundary seam check that replaces several broken native measurements.
	
	Section~\ref{sec:simulations} evaluates this relationship on matched clustered boundary-defect maps. We first test the measurement structure extracted from each upstream configuration and then repeat the calculation after adding only the proposed boundary-reconstruction rows. The separate restricted-set study remains an internal component ablation. This design compares the two constructions using the same extracted measurement structures without claiming that either framework globally outperforms the other.
	
	\section{Defect-adaptive lattice-surgery method}
	\label{sec:method}
	
	\subsection{Problem setting and overview}
	
	We consider lattice-surgery merge operations between two defect-adapted surface-code patches.
	In the defect-free setting, the target logical parity is obtained from a regular family of seam checks along the merge interface.
	With seam-boundary defects, however, this regular picture can fail:
	a native seam check may be broken, shifted, reduced, or replaced by an inferred effective check.
	The seam may therefore cease to be a uniform strip of directly measured operators.
	
	The central problem addressed in this work is the following:
	given a requested merge operation on two irregular defect-adapted patches, determine whether the desired joint logical parity can still be realized from the seam-related measurements that remain physically available, and if so, construct the corresponding explicit parity-extraction rule.
	
	Our method separates this task into three layers.
	First, we identify the set of defect-adapted seam measurements that remains available on the irregular geometry.
	Second, we determine whether the requested logical parity lies in the row space generated by the seam measurements and the pre-merge stabilizer constraints.
	Third, we convert the successful seam relation into an explicit post-processing rule over the measured outcomes.
	Under this viewpoint, irregular lattice surgery becomes a defect-adaptive parity-synthesis problem rather than a purely geometric seam-drawing problem.
	
	Although the examples in this paper focus on a two-patch merge, the synthesis step is not tied to a particular local defect pattern. Its inputs are an admissible seam-measurement set, the pre-merge stabilizer-constraint matrix $H_P^{\rm sep}$, and the target representative $\bm\ell$. Its output is either an explicit parity-extraction rule or a dual infeasibility certificate. The same representation can be used for other defect-adapted parity measurements whenever their available measurements and raw-outcome rules admit a binary-support description.
	
	\subsection{Defect-adapted seam construction}
	
	The figures in this subsection are mechanism-level schematics rather than members of the statistical defect ensemble. Exact binary supports for the success and failure witnesses are given in Appendix~\ref{app:a47a57b41_example}, whereas ensemble-level statistics are reported separately in Section~\ref{sec:simulations}.
	
	The geometric meaning of the proposed method is illustrated in Figs.~\ref{fig:data_defect_merge} and~\ref{fig:ancilla_defect_merge}.
	The two figures are deliberately chosen to represent two qualitatively different defect-adaptive seam mechanisms.
	Figure~\ref{fig:data_defect_merge} illustrates the support-changing case: a boundary data-qubit defect directly damages the native seam region and forces a redefinition of the merge on a boundary-deformed patch.
	Figure~\ref{fig:ancilla_defect_merge} illustrates the support-preserving case: a seam-check ancilla defect leaves the intended data-qubit support intact but removes one local measurement resource, so the seam value must be reconstructed through an effective measurement primitive.
	
	\begin{figure*}[t]
		\centering
		
		\begin{minipage}[b]{0.49\textwidth}
			\centering
			\includegraphics[width=\linewidth]{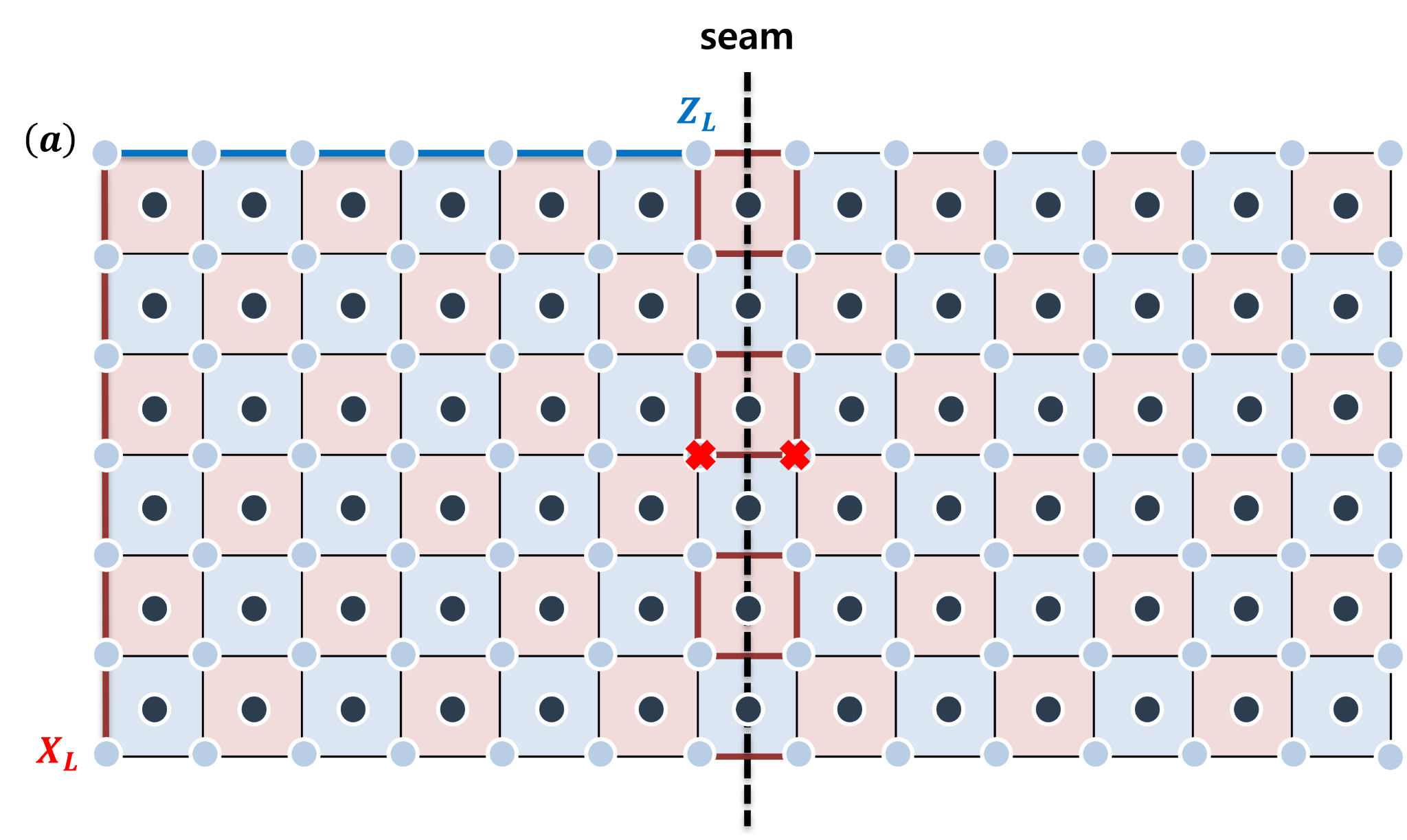}
		\end{minipage}
		\hfill
		\begin{minipage}[b]{0.49\textwidth}
			\centering
			\includegraphics[width=\linewidth]{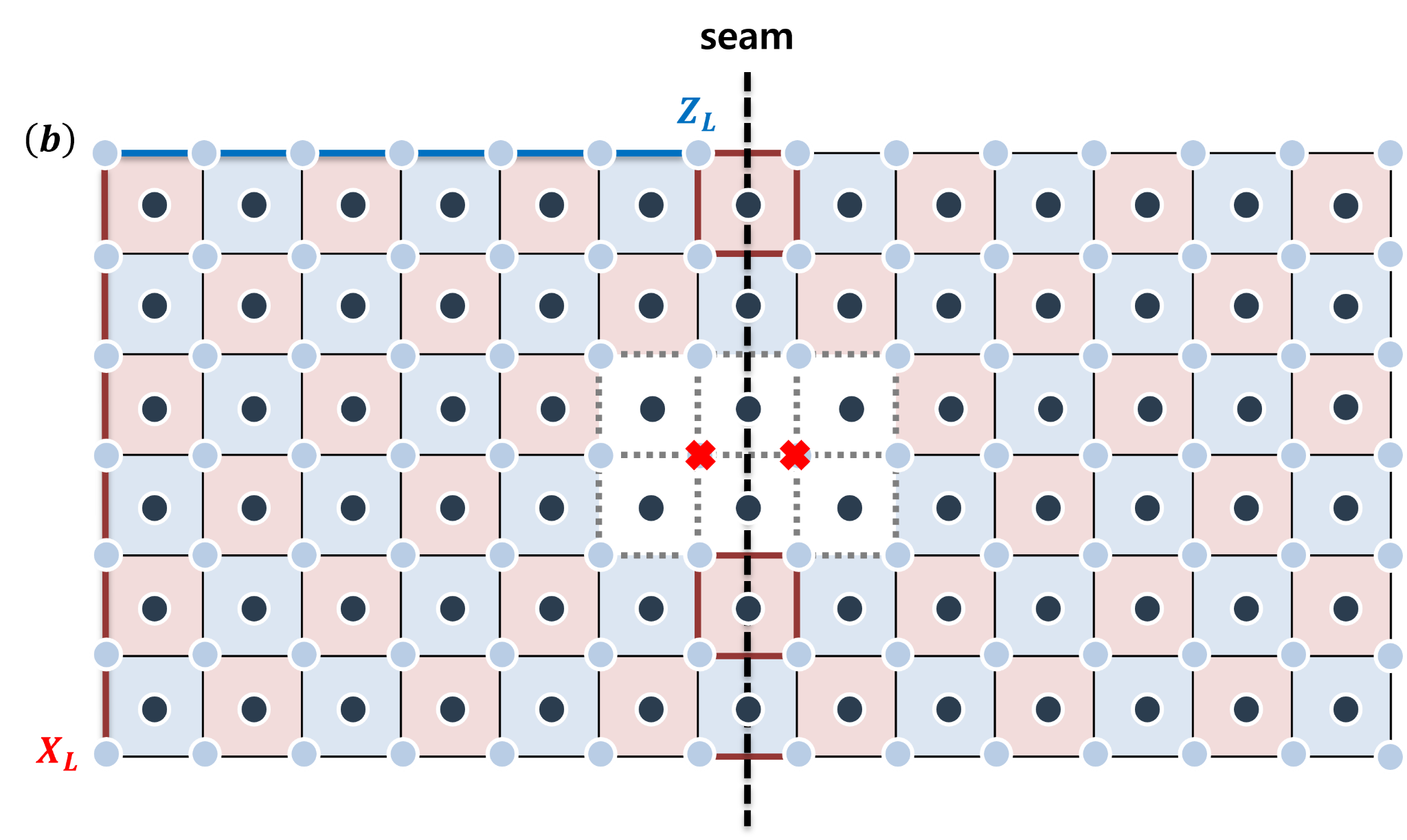}
		\end{minipage}
		
		\vspace{0.6em}
		
		\begin{minipage}[b]{0.49\textwidth}
			\centering
			\includegraphics[width=\linewidth]{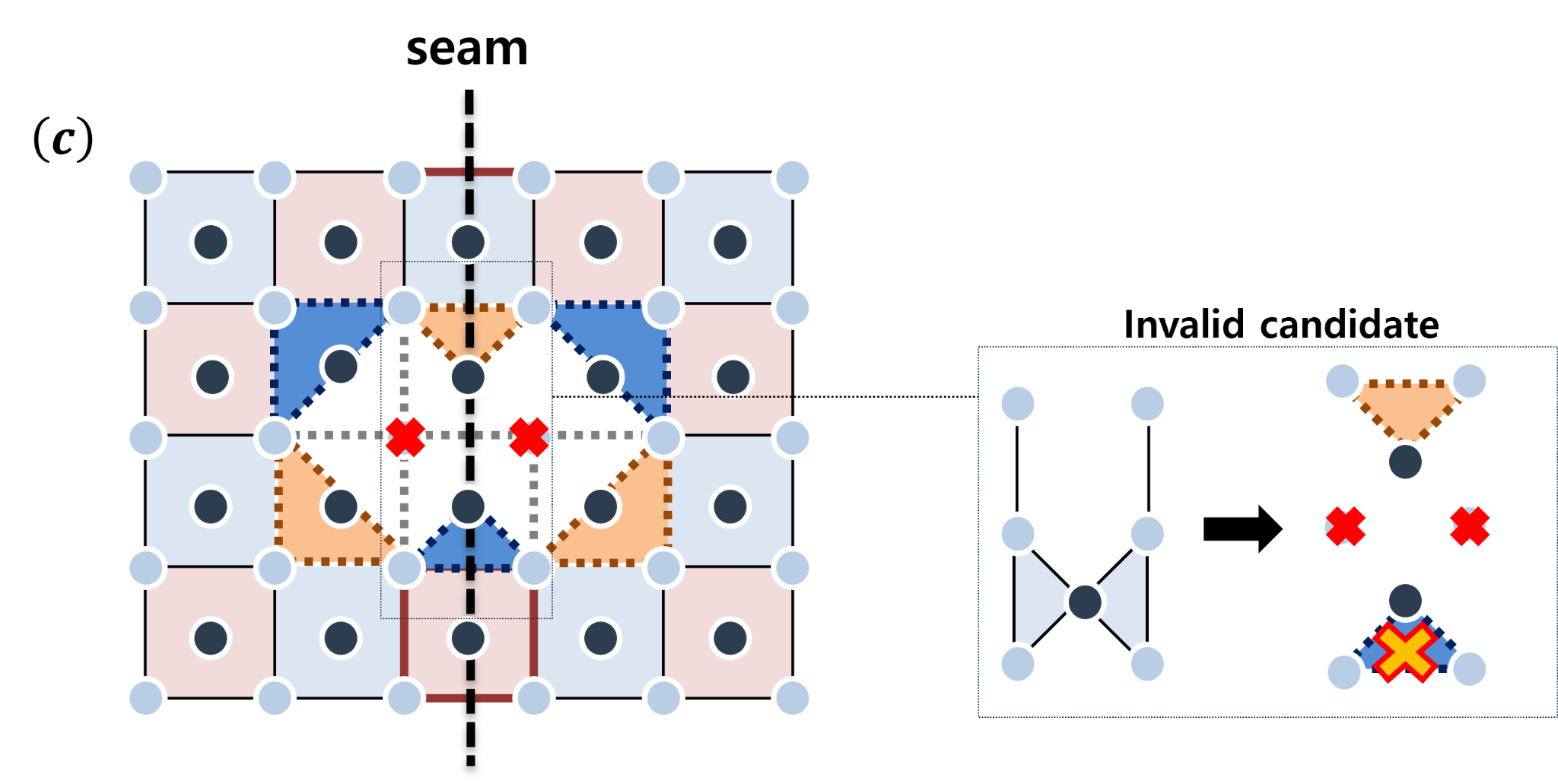}
		\end{minipage}
		\hfill
		\begin{minipage}[b]{0.49\textwidth}
			\centering
			\includegraphics[width=\linewidth]{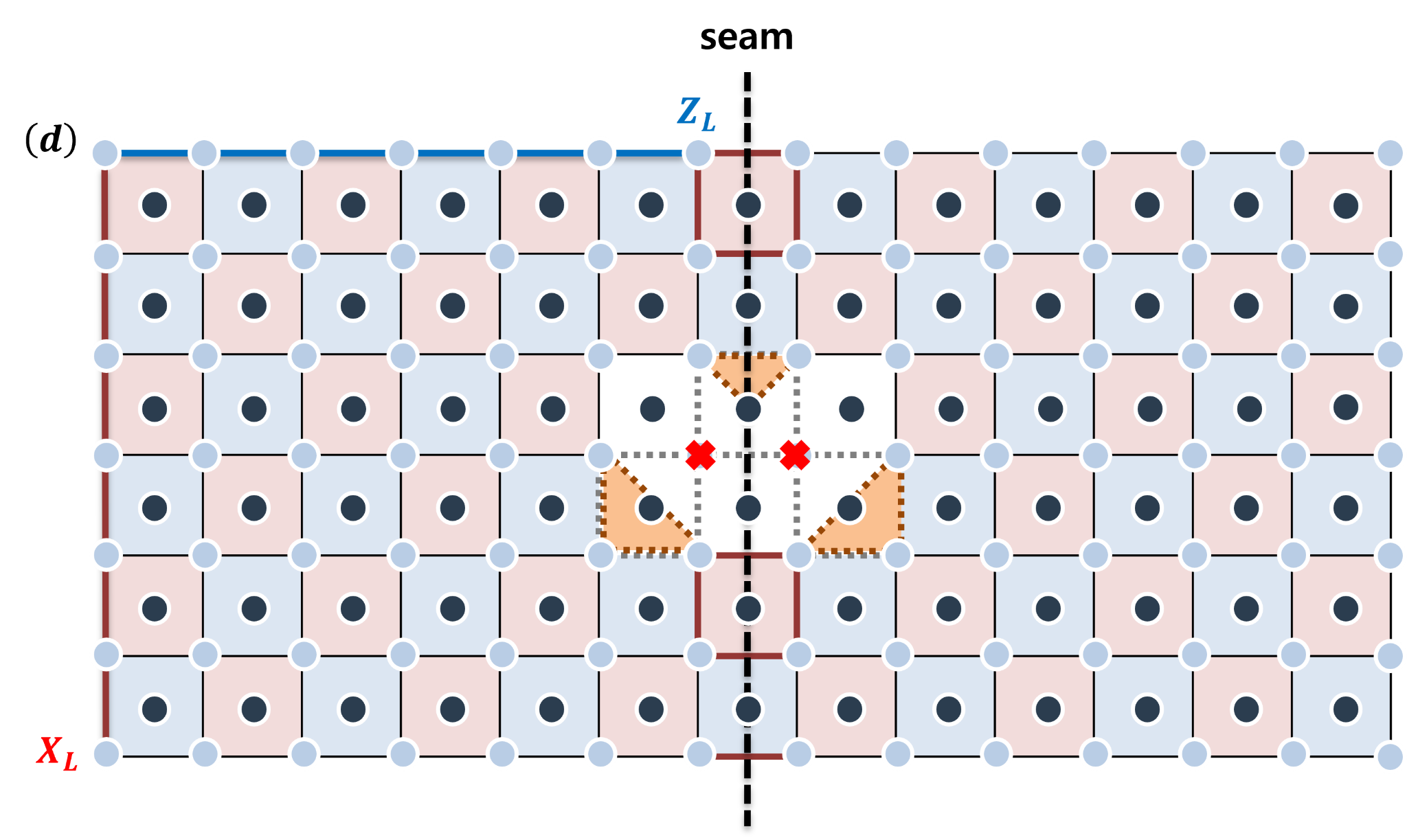}
		\end{minipage}
		
		\vspace{0.6em}
		
		\begin{minipage}[b]{\textwidth}
			\centering
			\includegraphics[width=\linewidth]{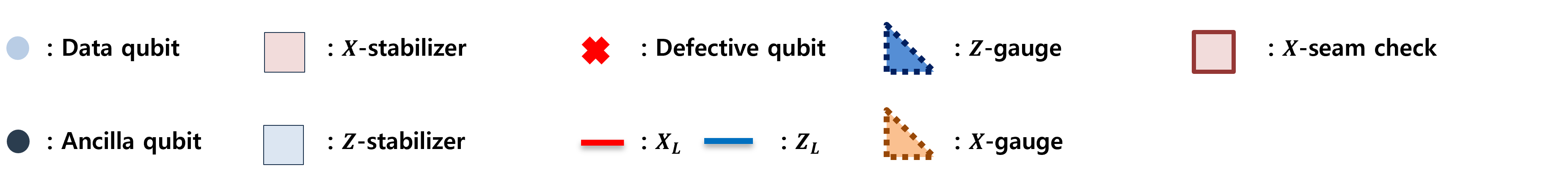}
		\end{minipage}
		
		\caption{%
			\textbf{Boundary data-qubit defects during an $X_L \otimes X_L$ merge.}
			(a) A representative seam-boundary defect pair destroys the native seam region of the intended merge, so the regular defect-free seam picture no longer applies.
			(b) The broken neighborhood is reconstructed in terms of locally available gauge fragments and defect-adapted effective checks.
			(c) A naive opposite-type completion is invalid because the required cross-patch gauge bridge is not measurable with the available resources. The surviving reduced opposite-type fragments are therefore each incompatible with the proposed $X$-type seam super-check.
			(d) The merge must therefore be carried out on the boundary-deformed patch using a defect-adapted seam-measurement set.
			The logical parity is synthesized from the seam-related measurements available on the deformed geometry rather than read off from the textbook seam strip. This is a mechanism-level schematic; statistical validation is reported separately in Figs.~\ref{fig:internal_ablation} and~\ref{fig:snl_integration}.
		}
		\label{fig:data_defect_merge}
	\end{figure*}
	
	We begin with the boundary data-qubit defect case in Fig.~\ref{fig:data_defect_merge}.
	Here the seam is damaged at the level of data-qubit support.
	As a result, one or more seam-adjacent checks are destroyed together with the local commuting structure that supported the textbook merge.
	In this situation, the desired parity can no longer be identified with the regular defect-free seam strip.
	Instead, the broken neighborhood must first be re-expressed in terms of locally available gauge fragments and defect-adapted effective checks, as illustrated in Fig.~\ref{fig:data_defect_merge}(b).
	
	This reconstruction is algebraically necessary, not merely a change of geometric representation.
	A geometrically plausible candidate seam measurement can still be invalid if the corresponding opposite-type completion fails the local commutation test.
	This is precisely the role of Fig.~\ref{fig:data_defect_merge}(c): the measurement is not rejected simply because of an abstract parity mismatch, but because the necessary cross-patch opposite-type bridge is not physically or algebraically available.
	Once that bridge is excluded, the remaining reduced opposite-type fragments are each incompatible with the proposed $X$-type seam super-check, so the measurement must be rejected.
	The merge is therefore redefined on the boundary-deformed patch and carried out using the admissible seam-measurement set of Fig.~\ref{fig:data_defect_merge}(d).
	A coordinate-level worked realization of this rejection-and-acceptance step is deferred to Appendix~\ref{app:a47a57b41_example}.
	
	This interpretation is consistent with the broader defect-adaptive literature.
	For boundary-touching defects, the appropriate viewpoint is that the code after deformation should be regarded as an ordinary surface code with a deformed boundary, rather than as the original code with a missing local measurement.
	In particular, when an $X$-boundary is deformed, the $Z$-distance need not be reduced, and conversely for a deformed $Z$-boundary~\cite{Siegel2023}.
	For our purposes, this means that the logical meaning of the merge must be read from the deformed seam geometry itself.
	Accordingly, Fig.~\ref{fig:data_defect_merge} should be understood as a seam-geometry reconstruction problem: the support of the effective seam-measurement set changes, while the intended logical parity class is preserved.
	
	\begin{figure*}[t]
		\centering
		
		\begin{minipage}[b]{0.49\textwidth}
			\centering
			\includegraphics[width=\linewidth]{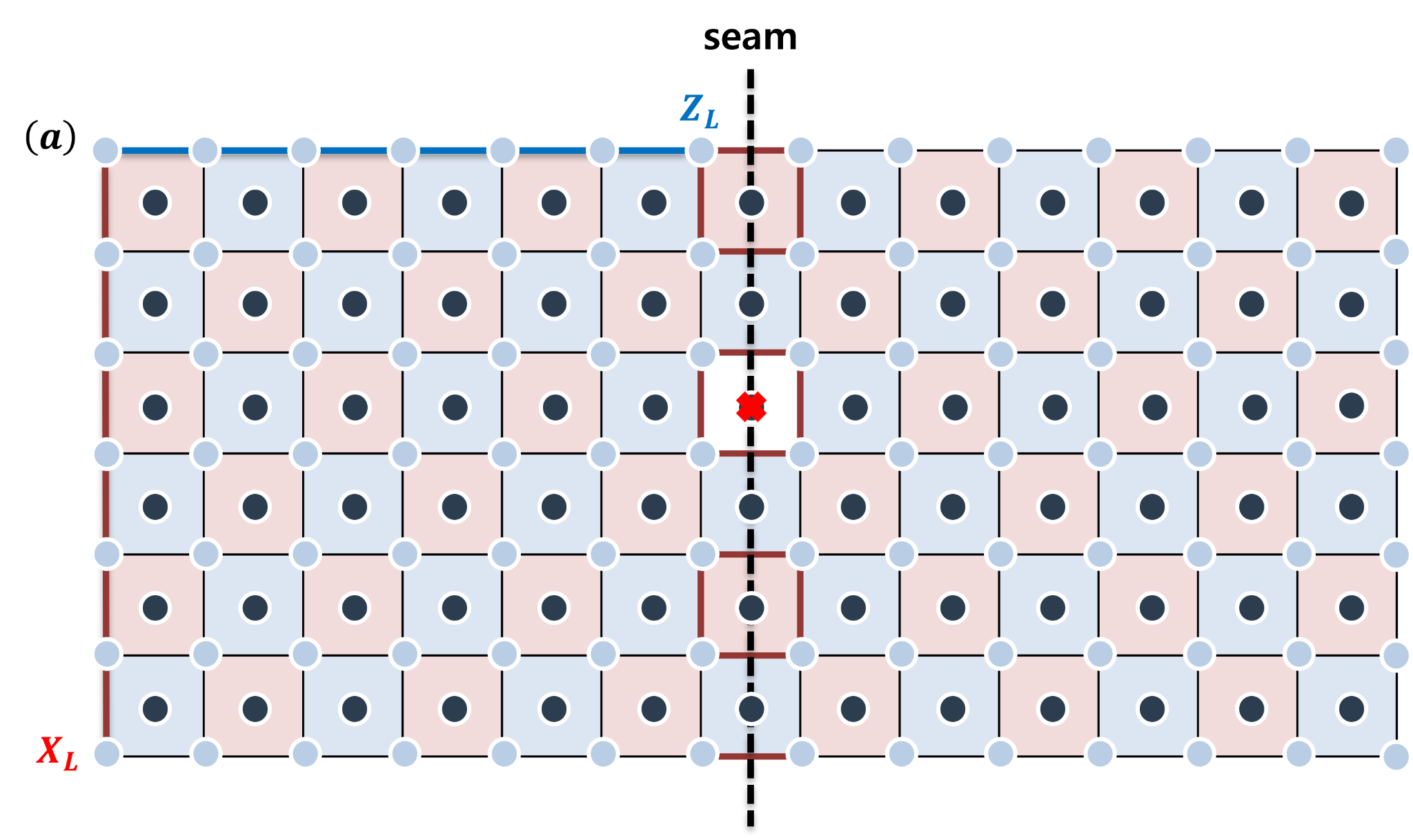}
		\end{minipage}
		\hfill
		\begin{minipage}[b]{0.49\textwidth}
			\centering
			\includegraphics[width=\linewidth]{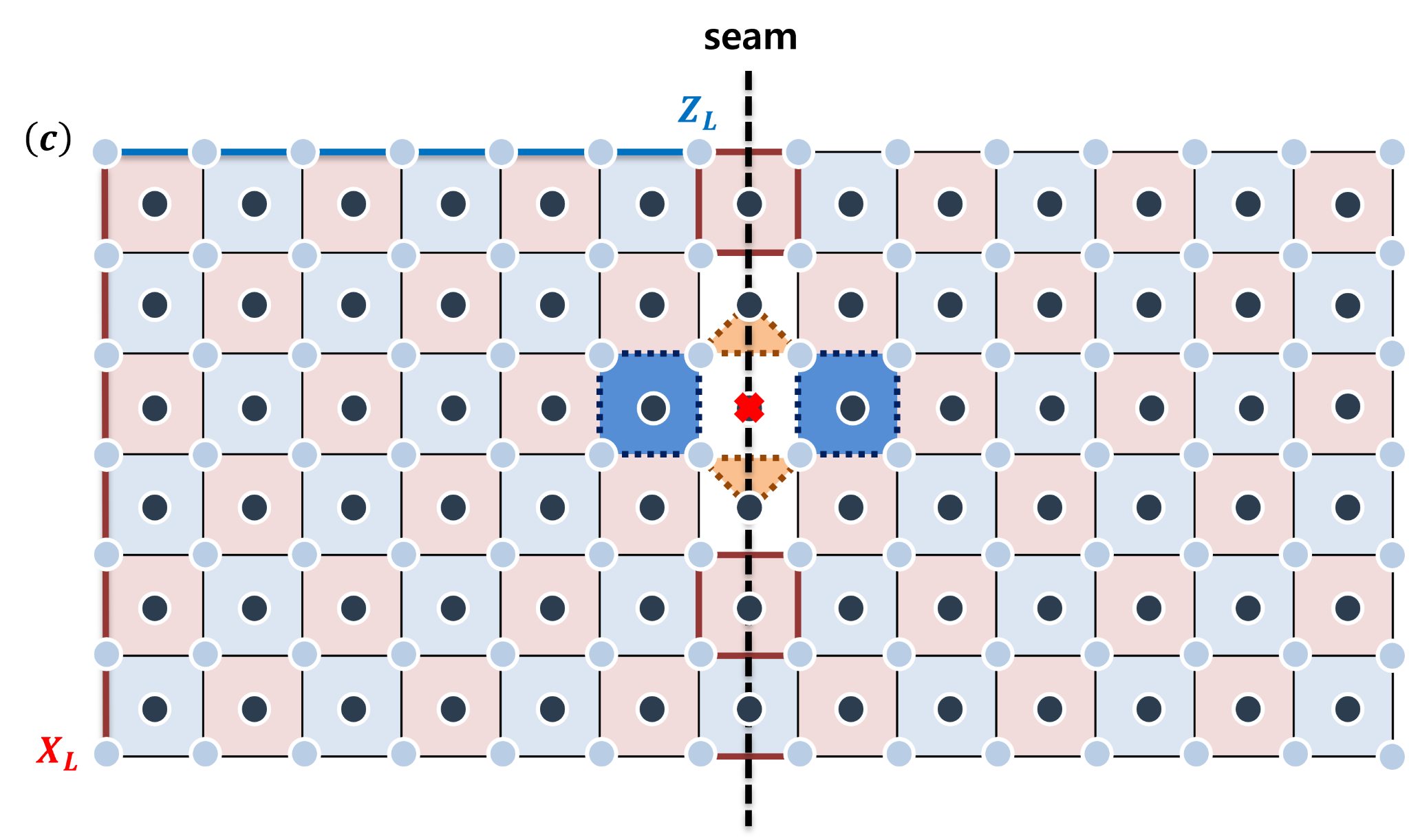}
		\end{minipage}
		
		\vspace{0.6em}
		
		\begin{minipage}[b]{0.62\textwidth}
			\centering
			\includegraphics[width=\linewidth]{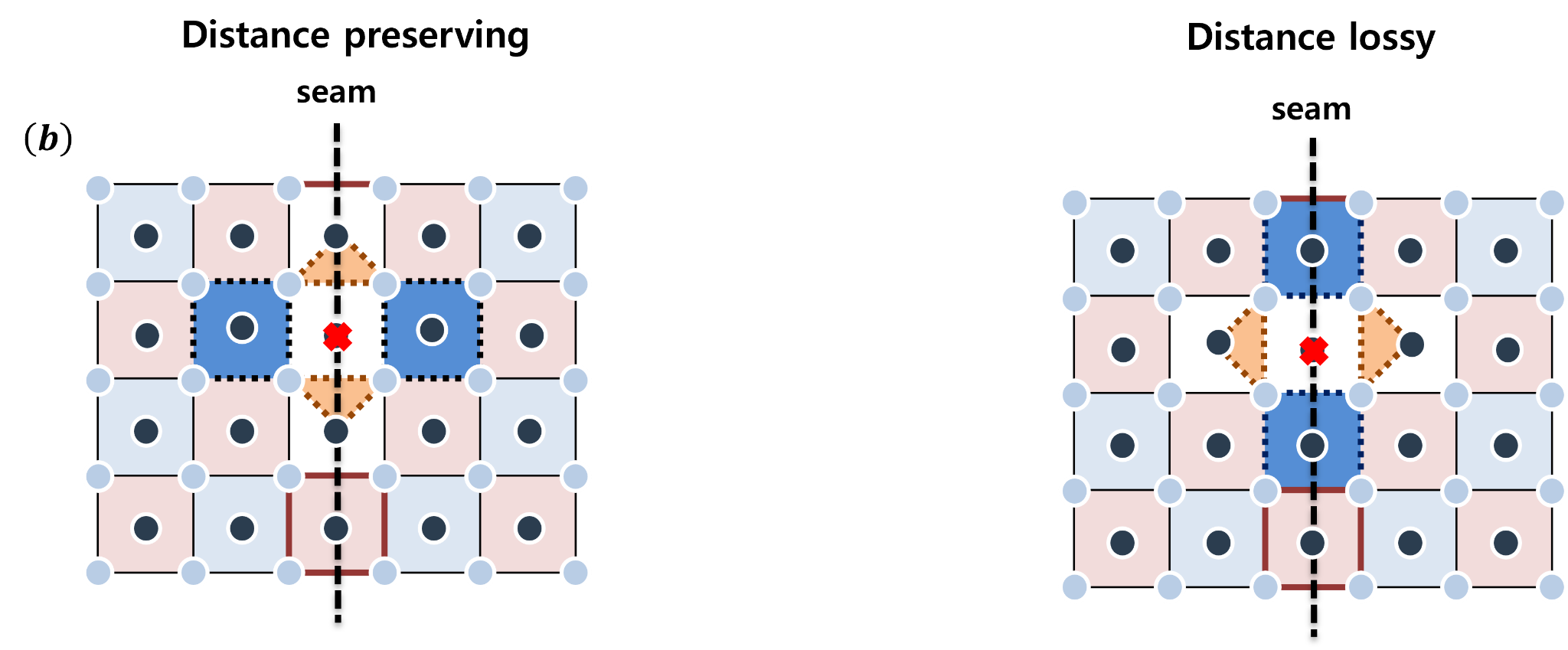}
		\end{minipage}
		
		\caption{%
			\textbf{Seam-check ancilla defect and support-preserving seam reconstruction.}
			(a) A defective seam ancilla removes one native seam measurement from the merge interface while leaving the intended data-qubit support intact.
			(b) The missing seam measurement is replaced by a pair of repurposed local gauge measurements whose product reconstructs the seam value.
			(c) The adjacent opposite-type larger checks are promoted to gauges, and their product defines an induced super-stabilizer that restores the local commutation structure.
			In contrast to the boundary data-qubit defect in Fig.~\ref{fig:data_defect_merge}, this case is handled primarily by support-preserving measurement reconstruction rather than by a larger boundary deformation.
			Different repurposing orientations can have different distance consequences in prior ancilla-repurposing work; here this information is used only as a local selection filter and not as a proof of the global distance of the proposed compiled operation. This is a mechanism-level schematic rather than a sampled statistical instance.
		}
		\label{fig:ancilla_defect_merge}
	\end{figure*}
	
	A complementary situation arises for the seam-check ancilla defect in Fig.~\ref{fig:ancilla_defect_merge}.
	Here the support of the intended seam check on the data qubits remains intact, but the ancilla that would directly measure that seam operator is unavailable.
	This is therefore not primarily a boundary-deformation problem.
	Instead, it is a measurement-reconstruction problem.
	Following the ancilla-repurposing construction of Ref.~\cite{SnL}, the missing native seam check is split into two lower-weight same-type gauge measurements carried out by neighboring resources.
	Their product reconstructs the seam value, while the adjacent opposite-type larger checks become gauges whose product defines an induced super-stabilizer~\cite{SnL}.
	In other words, the seam-check ancilla defect preserves the target support but changes the measurement primitive used to access it.
	
	This distinction is important for the interpretation of Fig.~\ref{fig:ancilla_defect_merge}.
	Panel (b) is the key operational step: the direct seam measurement is replaced by a pair of local gauge measurements.
	Panel (c) then shows the induced opposite-type constraint that restores local algebraic consistency.
	Thus, in contrast to Fig.~\ref{fig:data_defect_merge}, the seam-check ancilla defect does not require a support-changing seam reconstruction.
	Rather, it admits a support-preserving reconstruction in which the desired seam value is inferred from repurposed gauges and accepted only together with the corresponding opposite-type effective constraint.
	
	The ancilla-defect case also introduces an important distance consideration.
	Prior ancilla-repurposing analysis shows that two orientations are generally available around an isolated ancilla defect: one is distance-preserving, whereas the other is lossy and can reduce the relevant code distance by $2$ along the direction of the induced weight-8 super-stabilizer~\cite{SnL}.
	Near a boundary, this choice becomes especially important, because the wrong orientation can align the induced large check with the vulnerable logical direction and thereby shorten the minimum undetectable path.
	For this reason, our seam-measurement set excludes the orientation identified as locally lossy by that local analysis whenever the alternative local primitive is available.
	This rule is an admissibility choice motivated by a known local failure mechanism; it is not used below as a substitute for a shortest-logical-distance calculation of the complete proposed operation.
	
	Taken together, Figs.~\ref{fig:data_defect_merge} and~\ref{fig:ancilla_defect_merge} motivate the common abstraction used throughout this paper.
	The object employed in lattice surgery is not necessarily the textbook seam strip itself, but an \emph{effective seam-measurement set} defined on the actual irregular geometry and measurement schedule.
	Depending on the defect type, this set may include a boundary-deformed seam representative, reduced boundary checks, locally reconstructed seam segments, or gauge-inferred effective checks.
	The next two subsections convert this geometric picture into a certified and explicit parity-extraction rule: Section~\ref{sec:certified} formalizes the effective seam-measurement set in the parity-synthesis problem for the target logical parity, whereas Section~\ref{sec:execution} maps the successful seam relation to the actual gauge outcomes measured under the chosen schedule.
	
	\subsection{Admissible seam-measurement set}
	\label{sec:measurement_set}

	For a requested Pauli type $P\in\{X,Z\}$, we define the admissible seam-measurement set as
	\begin{equation}
		\mathcal M_P=\{(\bm e_j,\bm o_j,\sigma_j)\}_{j=1}^{m}.
		\label{eq:measurement_set}
	\end{equation}
	Here, $\bm e_j$ is the binary support of an effective $P$-type seam operator on the active data-qubit basis, $\bm o_j$ is the binary selector that reconstructs its outcome from the schedule-tagged raw measurement record, and $\sigma_j$ stores its source, orientation, and scheduling metadata. The effective-row and inference matrices used below are
	\begin{equation}
		E_P=\begin{bmatrix}\bm e_1\\ \vdots\\ \bm e_m\end{bmatrix},
		\qquad
		\widetilde O_P=\begin{bmatrix}\bm o_1\\ \vdots\\ \bm o_m\end{bmatrix}.
	\end{equation}

	An element is included in $\mathcal M_P$ only if its outcome is measured directly or inferred from a specified product of raw gauge outcomes available under the selected schedule. It must also commute with the retained opposite-type constraints, and all required ancillas, couplers, gauge fragments, and measurement rounds must be available. When two locally valid ancilla-repurposing orientations exist, the measurement set excludes the orientation identified by the local repair rule as avoidably lossy when the alternative is available. This filter constrains the available measurements but is not a global distance certificate. We emphasize that the algebraic certification is conditional on an admissible underlying measurement schedule. Although the schedule metadata $\sigma_j$ records how each effective measurement is realized, the GF(2) row-space test does not independently validate the physical measurement schedule or its hardware-level fault-tolerance properties.

	The pre-merge same-type stabilizer-constraint matrix $H_P^{\mathrm{sep}}$ and the independently chosen target logical representative $\bm\ell$ are not elements of $\mathcal M_P$. They are supplied separately to the subsequent parity-certification problem. This distinction prevents the target from being defined tautologically from the candidate seam measurements and separates the physical measurement-set construction from the subsequent binary linear algebra.

	The nontrivial contribution is therefore the construction of $\mathcal M_P$: deciding which native, gauge-inferred, fused-boundary, and repurposed measurements are physically available and specifying how their outcomes are obtained from the raw record. Once $\mathcal M_P$ is fixed, the GF(2) test below is sound and complete only for products of measurements contained in this set, modulo the pre-merge constraints in $H_P^{\mathrm{sep}}$. An infeasible result does not exclude a reconstruction based on additional hardware primitives or an enlarged measurement set.

\subsection{Certified seam-parity synthesis}
	\label{sec:certified}
	
	Once the defect-adapted seam-measurement set has been constructed, the next task is to determine whether it realizes the requested logical parity. For a requested merge of Pauli type $P \in \{X,Z\}$, we construct candidate $P$-type seam measurements and retain only those that pass the local commutation test against the residual opposite-type effective constraints and the relevant admissibility rules. For an $X_L\otimes X_L$ merge this means validating candidate $X$-type seam measurements against the retained defect-adapted $Z$-type constraints; for a $Z_L\otimes Z_L$ merge the same logic applies with $X$ and $Z$ exchanged.
	
	We collect the admissible effective seam rows into a matrix $E_P$ and the same-type constraints inherited from the separated pre-merge configuration into a matrix $H_P^{\mathrm{sep}}$.
	Let $\bm{\ell}$ denote a representative of the desired joint logical parity on the merged active-qubit set.
	The requested parity is generated by the chosen measurement set and the pre-merge constraints if there exist binary selector vectors $\bm{\alpha}$ and $\bm{\gamma}$ such that
	\begin{equation}
		\bm{\alpha}^{\top} E_P \oplus \bm{\gamma}^{\top} H_P^{\mathrm{sep}} = \bm{\ell},
		\label{eq:certified_synthesis}
	\end{equation}
	with arithmetic over GF(2). This is a standard linear feasibility test over the binary field~\cite{MacWilliams1977}. The following theorem should therefore be read as a compiler-correctness statement for a fixed admissible seam-measurement set, not as a new row-space criterion for stabilizer codes. Here, a \emph{parity-extraction rule} means a product of the effective seam rows admitted by the chosen seam-measurement set, followed by the raw gauge-outcome post-processing specified by the corresponding inference map.
	
	\begin{theorem}[Relative completeness within a fixed admissible seam-measurement set]
		\label{thm:certified_realizability}
		Fix an admissible seam-measurement set $\mathcal M_P$ with effective-row matrix $E_P$, a pre-merge same-type stabilizer-constraint matrix $H_P^{\mathrm{sep}}$, and a valid target joint logical representative $\bm{\ell}$. Using products of rows in this fixed admissible seam-measurement set, the requested joint parity can be extracted if and only if there exist binary selector vectors $\bm{\alpha}$ and $\bm{\gamma}$ satisfying
		\begin{equation}
			\bm{\alpha}^{\top} E_P \oplus \bm{\gamma}^{\top} H_P^{\mathrm{sep}} = \bm{\ell}.
		\end{equation}
		When such a solution exists, the selected effective seam rows differ from $\bm{\ell}$ only by pre-merge same-type stabilizer or super-stabilizer constraints. The seam selector $\bm{\alpha}$ also induces a raw-outcome selector
		\begin{equation}
			\bm{\beta}=\widetilde O_P^{\top}\bm{\alpha},
		\end{equation}
		so that the extracted merge parity is
		\begin{equation}
			s_{\mathrm{parity}}=\bm{\beta}^{\top}\bm{s}_P^{[1:T]}.
		\end{equation}
		The sign of this joint-parity measurement is then recorded as the corresponding Pauli-frame information.
	\end{theorem}
	
	\begin{proof}
		The ``if'' direction follows directly from Eq.~\eqref{eq:certified_synthesis}. The selected seam rows have binary support $\bm{\alpha}^{\top}E_P$, and the selected pre-merge constraints have support $\bm{\gamma}^{\top}H_P^{\mathrm{sep}}$. Since the rows of $H_P^{\mathrm{sep}}$ are already fixed in the separated code space, multiplying by them changes only the representative of the same logical parity class, not the logical parity being measured. The admissibility of every selected row of $E_P$ ensures compatibility with the retained opposite-type effective constraints. Mapping the selected effective seam rows through the inference matrix $\widetilde O_P$ gives the raw gauge selector $\bm{\beta}$ and therefore the parity-extraction rule.
		
		For the ``only if'' direction, fix the same admissible seam-measurement set. By assumption, any valid parity-extraction rule using this set is constructed from products of the available effective seam rows, modulo same-type constraints already fixed before the merge. Its binary support must therefore lie in the row space generated by $E_P$ and $H_P^{\mathrm{sep}}$, which is exactly Eq.~\eqref{eq:certified_synthesis}. Gaussian elimination over GF(2) is complete for this row-space membership test within the fixed measurement set.
		
		The theorem does not assert that infeasibility of Eq.~\eqref{eq:certified_synthesis} rules out every conceivable hardware-level reconstruction. It asserts that the requested parity is not generated by the admissible seam-measurement set supplied to the compiler. Enlarging the measurement set by adding additional measurement primitives, alternative repurposing choices, or longer-range fused checks may change the row space and can therefore change the feasibility result.
	\end{proof}
	
	\begin{proposition}[Dual infeasibility certificate]
		\label{prop:dual_failure}
		Let
		\begin{equation}
			M_P=
			\begin{bmatrix}
				E_P\\
				H_P^{\mathrm{sep}}
			\end{bmatrix}.
		\end{equation}
		If Eq.~\eqref{eq:certified_synthesis} has no solution, then there exists a binary vector $\bm w$ such that
		\begin{equation}
			M_P\bm w^{\top}=\bm 0,
			\qquad
			\bm\ell\bm w^{\top}=1.
			\label{eq:dual_failure_witness}
		\end{equation}
		Conversely, the existence of such a vector proves that $\bm\ell$ is not in the row space of $M_P$.
	\end{proposition}
	
	\begin{proof}
		The row space of $M_P$ and the null space of $M_P$ are orthogonal complements over GF(2). If $\bm\ell$ is not in the row space of $M_P$, separation by the orthogonal complement gives a vector $\bm w\in\ker M_P$ with $\bm\ell\bm w^\top=1$. Conversely, every vector in the row space of $M_P$ has zero inner product with every vector in $\ker M_P$, so Eq.~\eqref{eq:dual_failure_witness} excludes $\bm\ell$ from that row space.
	\end{proof}
	
	The vector $\bm w$ is a compact compiler certificate: every included seam row or pre-merge constraint row has even overlap with $\bm w$, while the requested target has odd overlap. As with Theorem~\ref{thm:certified_realizability}, this certificate is relative to the supplied measurement set and does not rule out a reconstruction obtained after adding new hardware primitives or additional measurements.
	
	Figure~\ref{fig:data_defect_merge} and Fig.~\ref{fig:ancilla_defect_merge} illustrate the two representative uses of this formulation.
	In Fig.~\ref{fig:data_defect_merge}, the native seam picture is support-changing, so $E_P$ must be built from the boundary-deformed seam-measurement set rather than from the defect-free seam strip.
	In Fig.~\ref{fig:ancilla_defect_merge}, the seam support is preserved but one seam row is replaced by an inferred effective row reconstructed from repurposed local gauges.
	In both cases, the certification step answers the same question: whether the available defect-adapted seam-measurement set still generates the requested logical parity modulo the pre-merge constraints.
	The certification and execution steps are summarized in Algorithm~\ref{alg:certified_seam_parity}; a coordinate-level worked example is given in Appendix~\ref{app:a47a57b41_example}.
	
	\subsection{Parity extraction and infeasibility certificates}
	\label{sec:execution}
	
	The output of the synthesis step is not merely the statement that a desired logical parity is realizable.
	What is ultimately needed for lattice surgery is an explicit parity-extraction rule specifying which measured bits must be combined to recover the merge outcome and the associated Pauli-frame update.
	
	Let $\bm{s}_P^{[1:T]}$ denote the measured gauge-outcome stream over the relevant schedule window, and let $\widetilde{O}_P$ denote the linear inference map from those raw outcomes to the effective seam-check outcomes.
	Once a valid seam selector $\bm{\alpha}$ has been found, the corresponding gauge-level selector is
	\begin{equation}
		\bm{\beta} = \widetilde{O}_P^{\top}\bm{\alpha},
	\end{equation}
	and the logical merge outcome is extracted as
	\begin{equation}
		s_{\mathrm{parity}} = \bm{\beta}^{\top}\bm{s}_P^{[1:T]}.
		\label{eq:executable_parity}
	\end{equation}
    
		In the boundary data-qubit defect case of Fig.~\ref{fig:data_defect_merge}, $\bm{\alpha}$ selects seam operators supported on the boundary-deformed patch.
	In the seam-check ancilla defect case of Fig.~\ref{fig:ancilla_defect_merge}, one seam row is implemented indirectly through repurposed local gauges, so $\bm{\beta}$ selects those raw gauge outcomes rather than a single native seam bit.
	This is also where the local orientation filter enters: the orientation identified by the local repair analysis as lossy is excluded when the alternative reconstruction is available. The filter constrains the measurement set but does not by itself determine the shortest logical operator of the complete compiled operation.
	\begin{algorithmblock}[alg:certified_seam_parity]{Certified seam-parity synthesis and raw-outcome extraction}
    
		\begin{algorithmic}[1]
			\Require Merge type $P$, candidate seam measurements, retained opposite-type constraints, pre-merge constraints $H_P^{\mathrm{sep}}$, target $\bm{\ell}$, and raw-outcome data
			\Ensure On success, selectors $\bm{\alpha}$ and $\bm{\beta}$ and the extracted parity; on failure, a dual witness $\bm w$
			\State Construct the admissible seam-measurement set $\mathcal M_P$, and assemble its effective-row matrix $E_P$ and inference matrix $\widetilde O_P$, retaining only measurements that are available under the schedule, commute with the opposite-type constraints, and pass the local orientation filter.
			\State Solve
			\[
			\bm{\alpha}^{\top} E_P \oplus \bm{\gamma}^{\top} H_P^{\mathrm{sep}} = \bm{\ell}
			\]
			over GF(2).
			\If{no solution exists}
			\State Find $\bm w\in\ker M_P$ such that $\bm\ell\bm w^\top=1$.
			\State \Return infeasible and the dual witness $\bm w$.
			\EndIf
			\State Extract the seam selector $\bm{\alpha}$ from the solution.
			\State Compute the gauge-level selector
			\[
			\bm{\beta} \leftarrow \widetilde{O}_P^{\top}\bm{\alpha}.
			\]
			\State Construct the parity-extraction rule
			\[
			s_{\mathrm{parity}} = \bm{\beta}^{\top}\bm{s}_P^{[1:T]}.
			\]
			\State \Return success, $\bm{\alpha}$, $\bm{\beta}$, and $s_{\mathrm{parity}}$
		\end{algorithmic}
	\end{algorithmblock}

	If Eq.~\eqref{eq:certified_synthesis} has no solution, the requested logical primitive is unavailable under the given defect-adapted seam-measurement set and pre-merge constraints.
	This should be regarded as non-realizability within the chosen seam-measurement set rather than as a patch-construction failure.
	Accordingly, we distinguish throughout this work among patch viability, parity-synthesis success, and downstream circuit-level execution.
	The resulting compiler procedure is summarized in Algorithm~\ref{alg:certified_seam_parity}.

	\subsection{Worked success and failure certificates}
	\label{sec:worked_output}

	Appendix~\ref{app:a47a57b41_example} now reports a compiler-generated certificate for the distance-$7$ three-defect boundary cluster rather than a target chosen from the product of the selected seam measurements. The union active-qubit basis has 95 coordinates, the pre-merge constraint matrix is nonempty with
	\begin{equation}
		H_X^{\mathrm{sep}}\in \mathrm{GF}(2)^{46\times95},
	\end{equation}
	and the target $\bm\ell\in\mathrm{GF}(2)^{95}$ is constructed from the deformed separated logical representatives before the seam-measurement set is assembled. Neither the seam rows alone nor the separated constraints alone generate the target:
	\begin{equation}
		\operatorname{rank}(E_X)=3<4=
		\operatorname{rank}
		\begin{bmatrix}E_X\\ \bm\ell\end{bmatrix},
	\end{equation}
	and
	\begin{equation}
		\operatorname{rank}(H_X^{\mathrm{sep}})=46<47=
		\operatorname{rank}
		\begin{bmatrix}H_X^{\mathrm{sep}}\\ \bm\ell\end{bmatrix}.
	\end{equation}
	The combined system is feasible with $\bm\alpha=(1,1,1)^\top$ and a weight-$10$ pre-merge constraint selector $\bm\gamma$, so the example genuinely uses both terms in Eq.~\eqref{eq:certified_synthesis}.

	The paired failure instance keeps the same active basis, the same $H_X^{\mathrm{sep}}$, and the same target $\bm\ell$, but removes one native seam-measurement row and all available reconstructions of that row. The rank then changes from
	\begin{equation}
		\operatorname{rank}(M_X^{\mathrm{fail}})=48
		\quad\text{to}\quad
		\operatorname{rank}
		\begin{bmatrix}M_X^{\mathrm{fail}}\\ \bm\ell\end{bmatrix}=49,
	\end{equation}
	and Sec.~\ref{app:certified_failure} provides an explicit dual witness. Thus the same independently fixed logical target changes from realizable to non-realizable within the evaluated measurement set solely because the available measurement set is reduced.

	\section{Compiler-level evaluation}
	\label{sec:simulations}

	The revised numerical study separates the compiler-level quantities directly evaluated in the main text from an archived geometry heuristic retained only for context. The main-text results report a controlled internal measurement-set ablation, upstream configuration generation, parity certification, and a stagewise detector-graph diagnostic. Appendix~\ref{app:auxiliary_diagnostics} retains a heuristic seam-geometry score as a descriptive compiler heuristic; it is not interpreted as a shortest-logical, detector-error-model, or end-to-end merge-circuit distance. The earlier memory-sensitivity plot is retained in the repository as archived exploratory material but is excluded from the manuscript because only aggregate results were preserved, preventing the reconstruction of valid confidence intervals in the low-failure-count regime.

	\begin{figure*}[t]
		\centering
		\includegraphics[width=0.90\textwidth]{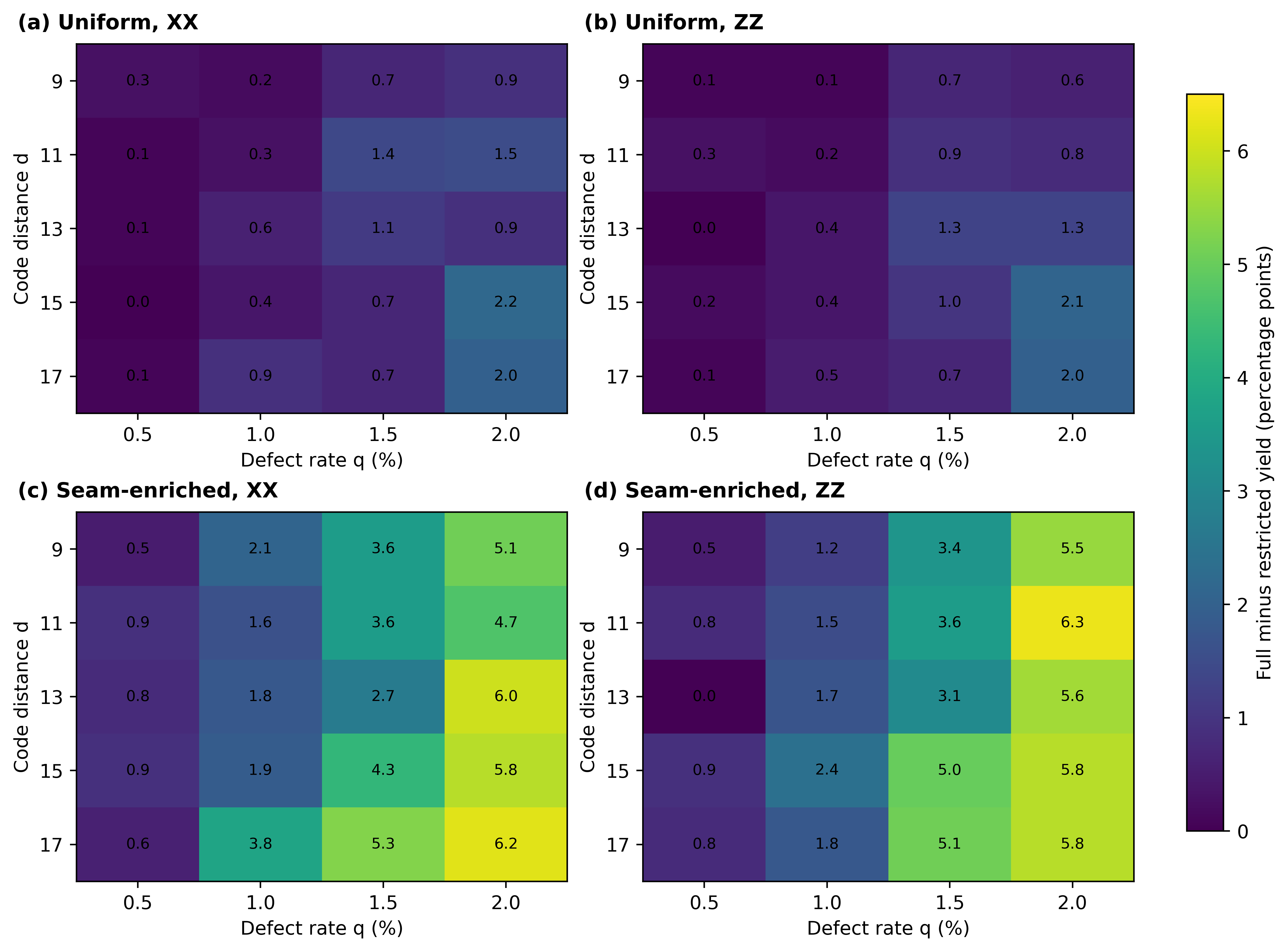}
		\caption{Controlled paired ablation of the admissible seam-measurement set. Each cell reports the full-set certification yield minus the restricted-set yield, in percentage points, on the same 1000 sampled maps. Panels (a) and (b) use the uniform Bernoulli ensemble; panels (c) and (d) use the seam-enriched ensemble. The restricted measurement set retains the same patch-admission and GF(2) certification layers but disables fused-boundary reconstruction and the local orientation filter. The figure is therefore a component ablation and not a comparison with a published method.}
		\label{fig:internal_ablation}
	\end{figure*}
	\subsection{Defect ensembles and statistical procedure}

	The paired internal ablation uses $d\in\{9,11,13,15,17\}$, $q\in\{0.5,1.0,1.5,2.0\}\%$, and 1000 maps per $(d,q,\text{ensemble},\text{merge type})$ point. In the \emph{uniform} ensemble, every data-qubit location in both patches and every seam-check-ancilla location is sampled independently with Bernoulli probability $q$. In the \emph{seam-enriched} ensemble, the bulk remains Bernoulli-$q$, while the two facing boundary columns, the transposed boundary rows used by the paired ZZ construction, and the seam-ancilla line are sampled at $q_{\rm seam}=\min(1,5q)$. The same physical map is evaluated using full and restricted measurement sets. Every binomial yield is accompanied by a Wilson 95\% confidence interval, and the paired difference is evaluated by bootstrap confidence intervals together with the both-success, full-only, restricted-only, and neither-success counts.

	The published-method integration uses a separate clustered data-qubit ensemble with seam-ancilla defects disabled. For each distance $d\in\{9,11,13,15,17\}$ and cluster parameter $\lambda\in\{0.5,1,2\}$, we sample 1000 independent maps. Bulk data defects are Bernoulli with probability $q_{\rm bulk}=0.002$. The number of boundary clusters is Poisson distributed with mean $\lambda$ (conditioned to be at least one for $\lambda>0$); each cluster has a uniformly sampled length from one to four boundary sites and is placed on the facing boundary of patch A, patch B, or both aligned boundaries, with aligned probability $0.2$. The identical physical defect map is evaluated for XX and ZZ merges. Each returned configuration is then converted to a union active-qubit basis over the merged and separated stages, and every extractable configuration is tested before the best certified configuration is selected. Missing support causes explicit rejection rather than silent truncation.

	\subsection{Controlled component-level ablation}

	Each internal instance is evaluated under two measurement sets designed to isolate the components introduced in this work. The \emph{restricted measurement set} uses the same patch-admission rules and the same GF(2) certification layer as the full method, but omits fused boundary reconstruction and the local orientation filter. The \emph{full measurement set} adds these two measurement-set construction rules. Table~\ref{tab:method_components} summarizes this controlled ablation; it is not a published competing method.

	\begin{table}[t]
		\centering
		\caption{Controlled internal ablation. Both variants use the same patch-admission and GF(2) certification layers.}
		\label{tab:method_components}
		\renewcommand{\arraystretch}{1.08}
		\begin{ruledtabular}
			\begin{tabular}{lcccc}
				Method & Patch & GF(2) & Fused & Orient. \\
				\hline
				Restricted measurement set & $\checkmark$ & $\checkmark$ & -- & -- \\
				Full measurement set & $\checkmark$ & $\checkmark$ & $\checkmark$ & $\checkmark$ \\
			\end{tabular}
		\end{ruledtabular}
	\end{table}

	Across the archived uniform scans, the paired gain of the full measurement set ranges from 0 to 2.2 percentage points for XX and from 0 to 2.1 percentage points for ZZ. The largest uniform gains occur at $d=15$ and $q=2\%$: $2.2$ percentage points for XX (paired-bootstrap 95\% CI $[1.4,3.1]$) and $2.1$ percentage points for ZZ (95\% CI $[1.3,3.1]$). In the seam-enriched scans, the maximum gains are $6.2$ percentage points for XX at $(d,q)=(17,2\%)$ (95\% CI $[4.7,7.8]$) and $6.3$ percentage points for ZZ at $(d,q)=(11,2\%)$ (95\% CI $[4.8,7.9]$). No archived map is certified by the restricted measurement set and rejected by the full measurement set. These values quantify the marginal contribution of the added measurement-set construction rules within the same compiler; they are not interpreted as a comparison against SnL, LUCI, or another published framework. Full pointwise Wilson and paired-bootstrap intervals are included in the released machine-readable tables.

	Figure~\ref{fig:internal_ablation} shows that the marginal gain is generally small in the uniform ensemble and becomes larger in the seam-enriched ensemble that directly exercises the proposed boundary and seam mechanisms. The largest paired gains are 6.2 percentage points for XX and 6.3 percentage points for ZZ, consistent with the interval estimates reported above.

	\begin{figure*}[t]
        \centering
        \includegraphics[width=0.96\textwidth]{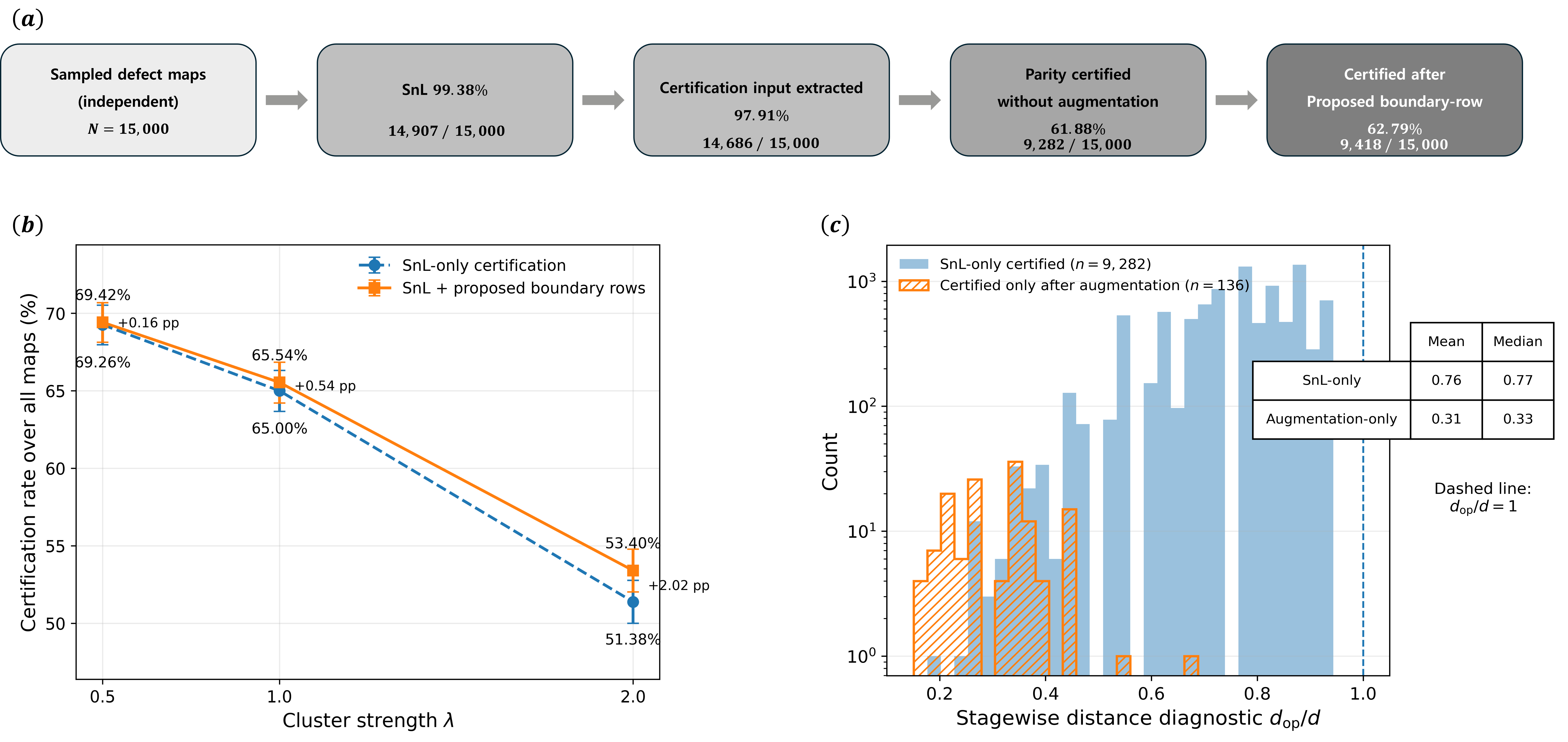}
        \caption{Certification of SnL configurations under clustered boundary defects. (a) Funnel from 15,000 independently sampled defect maps to adapted configuration output, GF(2)-input extraction, and certified joint parities. Percentages are unconditional over the sampled maps. (b) All-map certification rate versus boundary-cluster strength $\lambda$ before and after adding the proposed boundary-reconstruction rows; error bars are Wilson 95\% confidence intervals. The paired gains are 0.16, 0.54, and 2.02 percentage points for $\lambda=0.5,1,$ and $2$, respectively. (c) Distribution of the normalized stagewise detector-graph distance $d_{\rm op}/d$ for instances certified from the SnL configurations alone and instances certified only after adding the proposed rows. We define $d_{\rm op}$ as the minimum of the horizontal and vertical detector-graph distances over the merged configuration and the two separated subpatches. It is not an end-to-end merge-circuit fault distance. XX is shown; the paired XX and ZZ certification outcomes agree on 99.97\% of maps.}
        \label{fig:snl_integration}
    \end{figure*}

	\subsection{Certification of SnL configurations}
	\label{sec:snl_audit}

	To address the requested direct comparison with Leroux \emph{et al.}, we run the open-source SnL implementation on the 15,000 clustered boundary-defect maps described above and extract the adapted configurations returned by the implementation. For reproducibility, all SnL configurations used in this study were generated using the archived upstream SnL implementation at Git commit
\texttt{63ad64198c56a7d7d0021530707918c597671a0f}.
The same revision identifier is recorded in the released reproducibility package.
 For each configuration, we form $E_P$, $H_P^{\rm sep}$, and the target representative $\bm\ell$ from the adapted merged and separated patches and apply the same GF(2) certification test. We first apply the certification test using only the measurement structure extracted from that configuration and then repeat the test after adding the proposed boundary-reconstruction rows. This compares the two constructions using the same extracted measurement structures; it is not an end-to-end logical-error-rate comparison between complete merge circuits.

	As summarized in Fig.~\ref{fig:snl_integration}(a), the implementation returns at least one adapted configuration for 14,907 of 15,000 maps (99.38\%), and the certification inputs are extracted successfully for 14,686 maps (97.91\%). Using only the extracted measurement structures, the requested parity is certified in 9,282 maps (61.88\%). Adding the proposed boundary-reconstruction rows increases the certified set to 9,418 maps (62.79\%). The 136 additional certificates correspond to an overall paired increase of 0.91 percentage points, and no previously certified map is invalidated by adding the proposed rows.

	Panel~\ref{fig:snl_integration}(b) shows the unconditional certification rates aggregated over the five distances at each cluster strength. For $\lambda=0.5$, $1$, and $2$, the calculation without the proposed rows certifies 69.26\%, 65.00\%, and 51.38\% of maps, respectively, while adding the proposed boundary-reconstruction rows increases these values to 69.42\%, 65.54\%, and 53.40\%. The paired gains are therefore 0.16, 0.54, and 2.02 percentage points. We interpret this increase as a small, topology-dependent extension of certification coverage rather than a uniform improvement over the upstream method.

	For the stagewise detector-graph distance in Fig.~\ref{fig:snl_integration}(c), let $(d_h^{M},d_v^{M})$, $(d_h^{A},d_v^{A})$, and $(d_h^{B},d_v^{B})$ denote the horizontal and vertical detector-graph distances returned for the merged configuration and the two separated subpatches. We report
	\begin{equation}
		d_{\rm op}=\min\{d_h^{M},d_v^{M},d_h^{A},d_v^{A},d_h^{B},d_v^{B}\}.
	\end{equation}
	The set certified without the proposed rows has mean $d_{\rm op}/d=0.76$ and median $0.77$, whereas the set of instances certified only after adding the proposed rows has mean $0.31$ and median $0.33$, matching the production data used to regenerate panel~(c). Certification recovery therefore establishes algebraic operation availability within the enlarged measurement set but does not demonstrate restoration of the nominal fault distance. These results show that the methods are complementary at the certification stage: the upstream method supplies adapted configurations, while the proposed layer certifies the requested parity and adds a small number of physically valid measurements through boundary-reconstruction rows in specific clustered-defect topologies.

	\subsection{XX/ZZ scope and uncertainty}

	The same sampled physical maps are used for XX and ZZ. Under coordinate transposition, data and measurement locations are mapped bijectively while the checkerboard $X/Z$ assignment and the vertical-logical convention are exchanged. Consequently, the XX seam rows, separated $X$-type constraints, and target representative are mapped to the corresponding ZZ objects by simultaneous coordinate and row permutations of the GF(2) system. Finite-patch boundary conventions can nevertheless break exact numerical equality. The paired certification outcomes agree on 99.97\% of the 15,000 maps; the remaining five cases are retained in the released instance-level data rather than discarded. We therefore report XX in Fig.~\ref{fig:snl_integration} and treat ZZ as a paired consistency check, not as an assumed exact symmetry.

	All displayed yields use Wilson 95\% confidence intervals. Paired gains in the internal ablation are additionally accompanied by paired-bootstrap intervals. The stagewise-distance panel reports the full empirical distributions rather than only a mean. These uncertainty and distributional reports replace three-significant-figure point estimates without sampling information.

	\subsection{Scope of the numerical claims}

	The main-text evaluation does not claim a shortest-logical distance or a circuit-level logical-error rate for the proposed compiled operation, and it does not form a combined operational-success metric from quantities defined on different circuit contracts. The published-method study compares parity certification before and after the proposed rows are added to the same extracted measurement structures. Appendix~\ref{app:auxiliary_diagnostics} retains only the earlier heuristic geometry score as an explicitly labeled descriptive score. The GF(2) solve determines whether the requested target is generated by the evaluated rows, while the dual witness certifies non-membership in infeasible cases.

\section{Conclusion}
	\label{sec:conclusion}

	Defect-adaptive patch construction and defect-adaptive logical-operation certification are related but distinct tasks. Defect-adaptation and detecting-region frameworks construct irregular patches, measurement schedules, and detector structures on defective hardware. The present work addresses the downstream question of whether the rows supplied by such a construction generate a requested lattice-surgery joint parity and, if so, which raw outcomes must be combined to obtain it.

	The principal contribution is the admissible seam-measurement set, whose elements combine effective seam supports with raw-outcome selectors and schedule metadata. Once that set is fixed, standard GF(2) linear algebra tests whether it generates the target parity: a successful solve returns a raw-outcome selector, while an infeasible solve returns a dual witness proving that the target is absent from the row space generated by the evaluated measurement set. The conclusion is intentionally limited to the evaluated measurement set and does not exclude a reconstruction based on additional hardware primitives or a larger measurement set.

	The controlled internal ablation shows that fused-boundary and local orientation-filtering components increase requested-parity availability most clearly in seam-enriched ensembles. Applying the certification test to configurations generated by SnL provides the requested published-method comparison using the same extracted measurement structures. Adding the proposed boundary-reconstruction rows yields a small but statistically resolved increase in the certification rate under strongly clustered boundary defects, while the instances certified only after adding the proposed rows generally have reduced stagewise detector-graph distances. We therefore interpret the methods as complementary rather than as globally ranked competitors.

	The revised scope does not include an end-to-end shortest-logical-distance calculation or a circuit-level logical-error-rate estimate for the complete proposed merge. The auxiliary Appendix reports only a heuristic geometry score, which is not used as evidence for a physical-performance advantage. Constructing a time-dependent detector model for each irregular gauge schedule and evaluating its decoded logical failure rate remain important future work. Within its stated scope, the present certification layer converts irregular defect-adapted measurement resources into explicit lattice-surgery parity rules.

\begin{acknowledgments}
This work was supported by the Institute for Information \& Communications Technology Planning \& Evaluation (IITP) grant funded by the Korea government (MSIT) [No.~2020-0-00014, A Technology Development of Quantum OS for Fault-tolerant Logical Qubit Computing Environment]. This work was also supported by the Institute for Information \& Communications Technology Planning \& Evaluation (IITP)-ITRC (Information Technology Research Center) grant funded by the Korea government (Ministry of Science and ICT) (IITP-2026-RS-2021-II211810). This research was supported by Quantum Science and Technology Flagship Project(Quantum Computing)(No. RS-2025-25463854) through the National Research Foundation of Korea(NRF) funded by the Korean government (Ministry of Science and ICT(MSIT)).

\end{acknowledgments}
	
	\section*{Data and code availability}
    The data, code, and scripts used to reproduce the numerical results, regenerate the figures, and validate the compiler certificates are publicly available at \url{https://github.com/gsminsss/defect-adaptive-lattice-surgery}.

\appendix

\section{Auxiliary heuristic geometry score}
\label{app:auxiliary_diagnostics}

The main-text claims are based on certification yields and paired comparisons of the extracted measurement structures. For completeness, this Appendix retains one descriptive score from the archived simulation study. It characterizes a compiler heuristic under a simplified geometry contract and is not used as an end-to-end physical-performance claim.

\subsection{Heuristic seam-geometry score}

For each successfully compiled instance, the archived code assigned the same-type geometry score
\begin{equation}
 d_{\rm geo}=\max\!\left[1, d-\Delta_{\rm boundary}-\Delta_{\rm opp}-\Delta_{\rm orient}\right].
\end{equation}
Here $\Delta_{\rm boundary}$ is the rule-based penalty assigned to broken boundary rows, $\Delta_{\rm opp}$ is the penalty from opposite-type fallback rows, and $\Delta_{\rm orient}$ is the local penalty assigned when the locally preferred ancilla-repurposing orientation is unavailable. These integer penalties are compiler heuristics. Consequently, $d_{\rm geo}$ is neither a shortest logical operator, a detector-error-model distance, nor a circuit fault distance.

\begin{figure*}[t]
    \centering
    \includegraphics[width=0.96\textwidth]{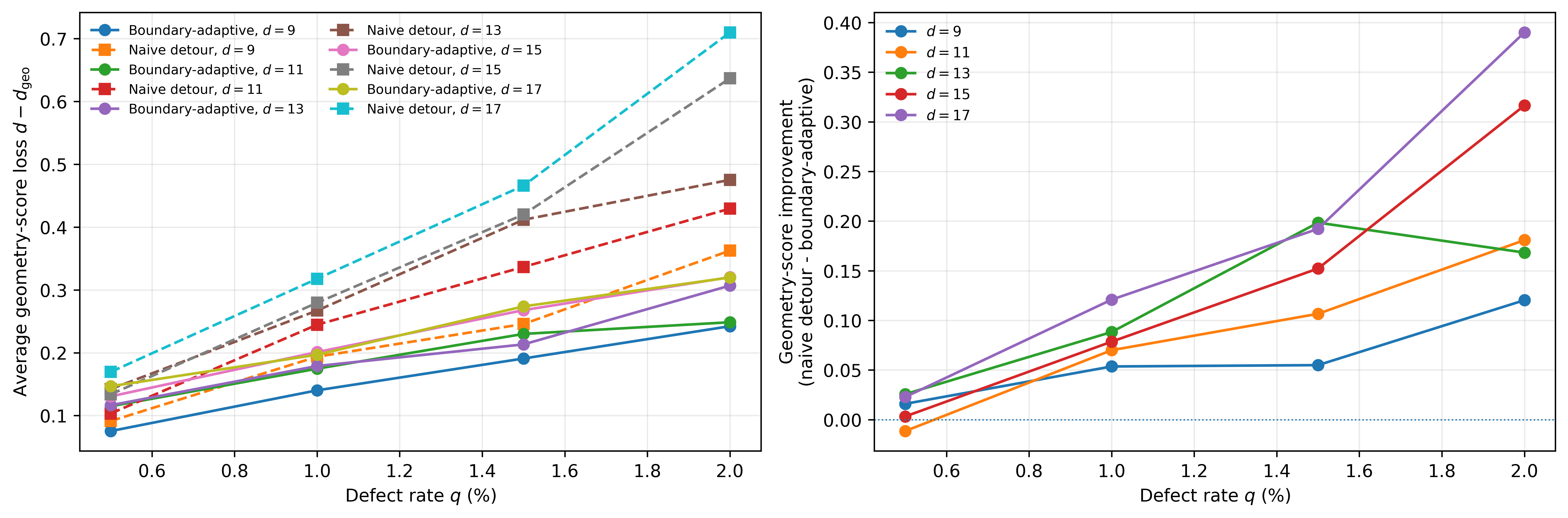}
    \caption{Archived heuristic seam-geometry score for the XX construction. The left panel compares the boundary-adaptive and naive boundary-detour heuristics through the average score loss $d-d_{\rm geo}$. The right panel reports the corresponding score improvement of the boundary-adaptive heuristic relative to the naive detour. This archived comparison is distinct from the full-versus-restricted certification ablation in Fig.~\ref{fig:internal_ablation}. The plotted score is not a calculation of the physical distance of an explicit irregular merge circuit. Because only aggregate means were retained for this archived scan, the figure is descriptive and is not used for a statistical-significance claim.}
    \label{fig:appendix_geometry_proxy}
\end{figure*}

Figure~\ref{fig:appendix_geometry_proxy} complements the exact
certification study only at the compiler-heuristic level. In particular,
it should not be combined with the stagewise detector-graph distance
$d_{\rm op}$ in Fig.~\ref{fig:snl_integration}(c), because the two
diagnostics are defined on different objects. The earlier proxy-based
memory-sensitivity output and its aggregate results are preserved under
\texttt{archive/legacy\_diagnostics/} in the repository but are not
included in the analysis because valid confidence intervals cannot be
reconstructed from the retained low-failure-count aggregates.

\section{Compiler-generated paired certificate for a three-defect boundary cluster}
\label{app:a47a57b41_example}

This appendix reports the actual binary certificate exported by the compiler for the distance-$7$ structured pattern denoted $A47/A57/B41$. In the implementation coordinates, the unavailable data qubits are
\begin{equation}
D=\{(13,7),(13,9),(15,7)\}.
\end{equation}
The compiler forms a union active-qubit basis
$\mathcal Q_{\rm union}=(q_1,\ldots,q_{95})$ covering the merged and separated stages. The complete coordinate ordering, matrices, and measurement ledger are included in the machine-readable file \texttt{certificate\_matrices.npz} and its accompanying metadata. The target $\bm\ell$ is generated from the deformed separated logical representatives before seam measurements are selected.

\subsection{Included seam rows and independently generated target}

For the successful XX certificate, the included seam matrix has three rows,
\begin{equation}
E_X=
\begin{bmatrix}
\bm r_1\\
\bm r_2\\
\bm r_3
\end{bmatrix}
\in\mathrm{GF}(2)^{3\times95}.
\end{equation}
Their physical supports are
\begin{align}
\operatorname{supp}(\bm r_1)
&=\{(13,1),(13,3),(15,1),(15,3)\},\\
\operatorname{supp}(\bm r_2)
&=\{(13,13),(15,13)\},\\
\operatorname{supp}(\bm r_3)
&=\{(11,7),(11,9),(13,5),(13,11),\nonumber\\
&\quad(15,1),(15,3),(15,13),(17,1),(17,3),\nonumber\\
&\quad(17,5),(17,7),(17,9),(17,11),(17,13)\}.
\end{align}
The first two rows are native surviving seam rows. The third is the fused boundary row spanning the two damaged seam windows.

The independently generated target has support
\begin{align}
\operatorname{supp}(\bm\ell)
=\{&(1,1),(1,3),(1,5),(1,7),(1,9),\nonumber\\
&(1,11),(1,13),(15,1),(15,3),(15,9),\nonumber\\
&(15,11),(15,13),(17,5),(17,7)\}.
\label{eq:full_target_support}
\end{align}
The pre-merge constraint matrix contains 46 rows,
\begin{equation}
H_X^{\rm sep}\in\mathrm{GF}(2)^{46\times95}.
\end{equation}
Crucially, the target is generated by neither input block alone:
\begin{align}
\operatorname{rank}(E_X)&=3,
&\operatorname{rank}\!\begin{bmatrix}E_X\\ \bm\ell\end{bmatrix}&=4,\\
\operatorname{rank}(H_X^{\rm sep})&=46,
&\operatorname{rank}\!\begin{bmatrix}H_X^{\rm sep}\\ \bm\ell\end{bmatrix}&=47.
\end{align}
This rules out the degeneracies identified by the referee: $H_X^{\rm sep}$ is nonempty, and $\bm\ell$ is not defined as the product of the selected seam measurements.

\subsection{Nondegenerate successful synthesis}
\label{app:certified_three_defect}

For
\begin{equation}
M_X=\begin{bmatrix}E_X\\H_X^{\rm sep}\end{bmatrix},
\end{equation}
the compiler obtains
\begin{equation}
\operatorname{rank}(M_X)=
\operatorname{rank}\!\begin{bmatrix}M_X\\\bm\ell\end{bmatrix}=49.
\end{equation}
One returned solution is
\begin{equation}
\bm\alpha=(1,1,1)^\top,
\end{equation}
with a weight-$10$ pre-merge constraint selector whose nonzero row indices are
\begin{equation}
\operatorname{supp}(\bm\gamma)
=\{1,2,3,4,24,25,26,27,28,30\}.
\end{equation}
Thus
\begin{equation}
\bm\alpha^\top E_X\oplus
\bm\gamma^\top H_X^{\rm sep}=\bm\ell.
\label{eq:full_success_certificate}
\end{equation}
Both terms are necessary: $\bm\alpha$ selects all three seam rows, while the nonzero $\bm\gamma$ changes their support to the independently fixed logical representative in Eq.~\eqref{eq:full_target_support}.

\subsection{Paired failure with the same target and pre-merge constraints}
\label{app:certified_failure}

We now keep the same 95-coordinate basis, the same $H_X^{\rm sep}$, and the same target $\bm\ell$, but remove the lower native seam row $\bm r_2$ and every available reconstruction of that measurement. The reduced matrix is
\begin{equation}
M_X^{\rm fail}
=\begin{bmatrix}
\bm r_1\\
\bm r_3\\
H_X^{\rm sep}
\end{bmatrix}.
\end{equation}
Binary elimination gives
\begin{equation}
\operatorname{rank}(M_X^{\rm fail})=48,
\qquad
\operatorname{rank}\!\begin{bmatrix}M_X^{\rm fail}\\\bm\ell\end{bmatrix}=49.
\end{equation}
An explicit dual witness is the 95-coordinate vector with support
\begin{align}
\operatorname{supp}(\bm w)
=\{&(1,1),(3,3),(5,5),(7,7),\nonumber\\
&(9,9),(11,11),(13,13)\}.
\end{align}
Direct verification gives
\begin{equation}
M_X^{\rm fail}\bm w^\top=\bm0,
\qquad
\bm\ell\bm w^\top=1.
\end{equation}
Therefore the same requested logical target that is realizable in Eq.~\eqref{eq:full_success_certificate} becomes non-realizable after only the specified measurement resource is removed. This is a statement relative to the evaluated measurement set and does not exclude a future reconstruction obtained by enlarging the measurement set.

\subsection{Raw-outcome selector}
\label{app:exec_three_defect}

The successful seam selector maps to a raw-measurement selector of weight 13. In the one-based raw-ledger ordering exported with the certificate,
\begin{equation}
\operatorname{supp}(\bm\beta)
=\{1,2,7,8,10,12,22,27,28,30,31,32,34\}.
\end{equation}
Equivalently, the extracted parity is the XOR of the two native seam outcomes and the following 11 local gauge outcomes used to infer the fused row:
\begin{align}
\mathcal G_{23}=\{&g_{(11,5),(11,7)},
 g_{(11,5),(13,5)},\nonumber\\
&g_{(11,9),(11,11)},
 g_{(11,11),(13,11)},\nonumber\\
&g_{(15,1),(15,3)},
 g_{(15,9),(15,11)},\nonumber\\
&g_{(15,9),(17,9)},
 g_{(15,11),(17,11)},\nonumber\\
&g_{(15,13),(17,13)},
 g_{(17,1),(17,3)},\nonumber\\
&g_{(17,5),(17,7)}\}.
\end{align}
Writing $s(\bm r_j)$ for a native seam outcome and $s(g)$ for a raw gauge outcome, the compiler returns
\begin{equation}
s_{X_LX_L}=s(\bm r_1)\oplus s(\bm r_2)
\oplus\bigoplus_{g\in\mathcal G_{23}}s(g).
\end{equation}
This completes the chain from an independently generated logical target, through a genuinely nondegenerate GF(2) certificate, to an explicit raw-outcome rule. The repository validation script reproduces the full matrix equality and the dual-witness checks.

	\bibliographystyle{apsrev4-2}
	\bibliography{sn-bibliography}
	
\end{document}